\documentclass[journal,twoside,web]{ieeecolor}
\usepackage{generic}
\usepackage{cite}
\usepackage{amsmath,amssymb,amsfonts}
\usepackage{algorithmic}
\usepackage{graphicx}
\usepackage{algorithm,algorithmic}
\usepackage{hyperref}
\hypersetup{hidelinks=true}
\usepackage{textcomp}
\def\BibTeX{{\rm B\kern-.05em{\sc i\kern-.025em b}\kern-.08em
    T\kern-.1667em\lower.7ex\hbox{E}\kern-.125emX}}
\usepackage{booktabs}
\usepackage{siunitx}
\usepackage{nicefrac}
\usepackage{comment}
\usepackage{cite}
\usepackage{siunitx}
\newtheorem{theorem}{Theorem}

\usepackage[dvipsnames]{xcolor}

\newlength{\mlLegendThickness}
\newlength{\mlLegendHeight}
\definecolor{dunkelblau}{rgb}{0.0, 0.2314, 0.6196}%
\definecolor{hellblau}{rgb}{0.000, 0.7451, 1.0000}%
\definecolor{rot}{rgb}{0.6980,0.1333,0.1333}%

\newcommand{\mycomment}[1]{}

\begin{document}
\title{Asymmetry Demystified: Strict CLFs and \\  Feedbacks for Predator–Prey Interconnections}
\author{Miroslav Krstic%
\thanks{M. Krstic is with Department of Mechanical and Aerospace Engineering, University of California San Diego, La Jolla, CA 92093-0411, USA (e-mail: mkrstic@ucsd.edu).} }

\maketitle

\begin{abstract}
The difficulty with control of population dynamics, besides the states being positive and the control having to also be positive, is the extreme difference in the dynamics near extinction and at overpopulated states. As hard as global stabilization is, even harder is finding CLFs that are strict, don't require LaSalle arguments, and permit quantification of convergence. Among the three canonical types of two-population dynamics (mutualism, which borders on trivial, predator-prey, and competition, which makes global stabilization with positive harvesting impossible), predator-prey is the ``sweet spot'' for the study of stabilization. Even when the predator-prey interaction is neutrally stable, global asymptotic stabilization with strict CLFs has proven very difficult, except by conservative, hard-to-gain-insight-from Matrosov-like techniques. 

In this little note we show directions for the design of clean, elegant, insight-bearing, majorization-free strict CLFs. They generalize the classical Volterra-style Lyapunov functions for population dynamics to non-separable Volterra-style constructions. As a bonus to strictification as an analysis activity, we provide examples of concurrent designs of feedback and CLFs, using customized versions of forwarding and backstepping (note that, in suitable coordinates, predator-prey is both strict-feedforward and strict-feedback), where the striking deviations from these methods' conventional forms is necessitated by the predator-prey's states and inputs  needing to be kept positive. 
\end{abstract}

% \begin{IEEEkeywords}
% Enter key words or phrases in alphabetical order, separated by commas. Using the IEEE Thesaurus can help you find the best standardized keywords to fit your article. Use the thesaurus access request form for free access to the IEEE Thesaurus: \underline{https://www.ieee.org/publications/services/thesaurus-acce}\\
% \underline{ss-page.com.}
% \end{IEEEkeywords}

\section{Introduction}

\IEEEPARstart{T}{his}
%modestly aiming 
succinct note, deliberately light on remotely related literature, is focused on pursuing inspiration from two references by the author: the decade-old \cite{malisoff2016stabilization} and the very recent \cite{11080060}. Both papers deal with global stabilization of positive-state positive-input systems in the style of predator-prey dynamics. Paper \cite{malisoff2016stabilization} deals with the architecture where only the predator is harvested, while \cite{11080060} deals with the situation where the simultaneous, non-discriminating harvesting of both the predator and the prey is inevitable. The two structures give rise to different challenges in stabilization. In both structures, the construction of Lyapunov functions that are strict is one of the key challenges. 

Achieving stabilization with a Lyapunov function whose time derivative is only negative semidefinite is at least moderately concerning from the engineering point of view, because a nonstrict Lyapunov function offers neither a reassurance of robustness to even small modeling errors nor any quantitative measure of a convergence rate. Mathematically, the nonstrictness of a Lyapunov function is even more unsettling---the converse Lyapunov theorem guarantees the existence of a strict function; how close is one to such a valuable function that surely exists, and how does one find it, at least in a particular case, and preferably systematically in more general situations?

Those are the drivers of our study. 

Biological population dynamics are undeniably a secondary interest of ours. But the primary interest is nonlinear control design. The predator-prey model is a benchmark for a difficulty in strict CLF construction for positive-state positive-input systems that has heretofore not been overcome.\footnote{Except in a manner that is far from satisfactory, aesthetically and conceptually, in a single result \cite{malisoff2016stabilization}, where the Matrosov means of strictification are so complex and so conservative that neither further interest nor replication has followed.} So, the reader should not seek biological fidelity and the capturing of details of predator-prey dynamics like open-loop limit cycles. Our goal is to present something as clean, transparent, and simple as possible, to chart a structural path for a possible further investigation of strict CLF-equipped stabilizer designs for actual biological (ecological, biotechnological, epidemiological, etc.), economic, and other population dynamics. 

The inspiration for our work here also comes from the CLF strictification for the SIR dynamics in \cite{9740520}. Our predator-prey dynamics are essentially the (S,I) portion of the SIR, where I\,=\,predator  and S\,=\,prey. However, the harvesting in this paper enters the model bilinearly, whereas the vaccination in \cite{9740520} enters additively. The strict CLFs we supply here are definitely distinct from the one in \cite[(38), (37)]{9740520} and lead to much simpler proofs, in a couple lines. 

{\bf What does the note {\em contribute}?} 
\begin{enumerate}
\item 
New ideas---for population dynamics---and perhaps other systems with positive states and inputs, for how to design CLFs that are {\em strict, non-Matrosov, and elegant}. These ideas draw on Lyapunov strictifications for mechanical systems, including even fluid mechanics. In our liquid-tank stabilization \cite{9744516},  we specifically combine the total energy \cite[(18)]{9744516}, which is a non-strict CLF,  with another backstepping-inspired non-strict Lyapunov functional \cite[(19)]{9744516}, for strictness of the composite Lyapunov functional. 

Of course, in populations, the challenges are  different than in fluids or robotics, due to the positivity of all the species states and the bilinear nature of the coupling of species. In particular, the backstepping transformation we propose for the synthesis of the second Lyapunov function takes the form of a {\em ratio} of population concentrations. And the Lyapunov augmentations that achieve crossterm cancellations in the derivative of the secon Lyapunov function not only vastly differ from trigonometric nonlinearities in robotics or quadratic nonlinearities for fluids but also from the {\em separable} Volterra nonlinearities encountered in the existing population dynamics literature. 

To summarize, the CLFs we offer differ from what is available in the literature---they are, all at once, (1) non-Matrosov, (2) non-separable, (3) explicit,  (4) insight-bearing. 

%It is premature in the Introduction to get into the details of the ideas proposed, and they are summarized in the Conclusions, but let us say informally that the approach proposed supplements a standard non-strict Lyapunov construction with an additional Lyapunov function that exposes dissipation of a state not seen by the first candidate. This is achieved by two augmentations: by introducing a ratio of population concentrations---playing a role akin to a backstepping transformation---and by adding a special crossterm-canceling function of the state whose dissipation is missed by the basic Lyapunov function. 

\item While the note's principal objective is not control design but primarily CLF strictification under prespecified controllers, we do explore the concurrent control and CLF design, and elucidate that as hard as global stabilization is with positive-state restrictions, stabilization with  the positive-input restriction is much harder. We design CLFs and then show how they ``break'' when control is restricted to be positive. And then, finally, we demonstrate the CLF design under positive harvesting to actually be possible---with unconventional backstepping. 
\end{enumerate}

If this paper's ideas do carry general {potential}---which only future will tell---they may permit generalizations, mentioned in the Conclusions. The opportunities are beyond population dynamics: for systems with positive states and controls,  bilinear in the input, and  with severe asymmetries in dissipation. 

Indeed, the paper is written primarily for nonlinear control researchers, using a population model as a benchmark. 

\underline{Organization.} Following the problem formulation in Section \ref{sec-formulation}, the CLF constructions under predator-only and simultaneous harvesting structures are given, respectively, in Sections \ref{sec-sf} and \ref{sec-both}. 
%Since backstepping-style coordinate changes are crucial in this approach to strictification, i
Integrator forwarding and backstepping, as approaches for simultaneous CLF and feedback design, are explored in a ``bonus'' Section \ref{sec-bonus}. The majority of the globally strict CLFs are accompanied by feedback laws that employ negative values on a  fraction of the state space, when the harvested predator is in deficit relative to the prey. While this  illustrates the difficulty entailed in designing a CLF that is admissible only for positive inputs, we present one ``small triumph'' of a global strict-CLF+positive-feedback construction.

\section{Problem Formulation}
\label{sec-formulation}

We study two models,
\begin{subequations}
\label{pp-sf}
\begin{eqnarray}
\label{pp-sf1}
\dot X &=& (1-Y)X
\\
\label{pp-sf2}
\dot Y &=& (X-U)Y
\end{eqnarray}
\end{subequations}
and
\begin{subequations}
\label{pp-both}
\begin{eqnarray}
\label{pp-both1}
\dot X &=& (2-Y-U)X
\\
\label{pp-both2}
\dot Y &=& (X-U)Y\,,
\end{eqnarray}
\end{subequations}
where, in both models, $X\geq 0$ is the concentration of the prey, $Y\geq 0$ the concentration of the predator, and $U\geq 0$ the harvesting (removal) {\em rate}. In both models, the prey feeds, and reproduces, on a constant external resource whose strength is 1, and the predator feeds and reproduces on the prey concentration $X$. 

The key difference in the models is that in the model \eqref{pp-sf} the control is able to harvest the predator only, whereas in the model \eqref{pp-both} inevitably harvests both populations and, intuitively, to achieve a stable balance at non-extinct equilibrium, has to be ``more cautious.''

Among the several reasons for choosing the benchmark systems \eqref{pp-sf} and \eqref{pp-both} as we do, one of the reasons is that, for a constant input $U=1$, the two systems are identical. They are both $\dot X = (1-Y)X, \dot Y = (X-1)Y$, have $X=Y=1$ as a globally stable equilibrium in the first quadrant, fail to converge to the equilibrium (i.e., they require stabilizing feedback), and their trajectories flow along the invariants 
\begin{equation}
\label{eq-invariants}
X + Y - \ln(XY) = X_0+Y_0 - \ln(X_0 Y_0)\,, 
\end{equation}
displayed in Figure \ref{fig:open-loop}. 

\begin{figure}[t]
\centering
\includegraphics[width=.7\columnwidth]{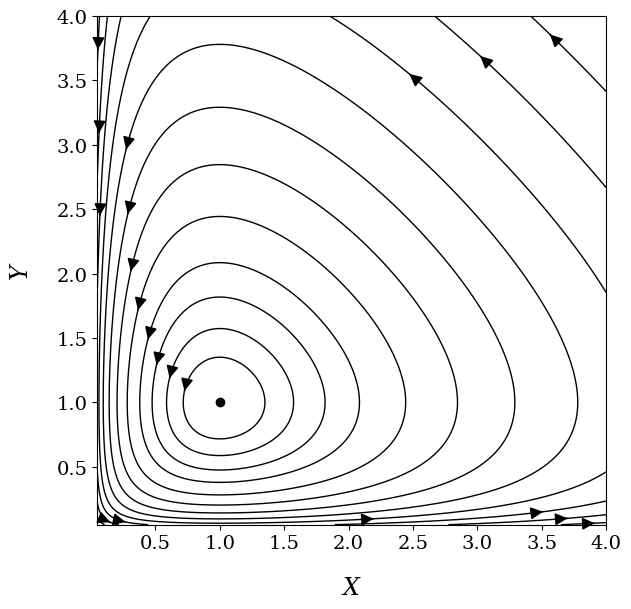}
\caption{$V(X,Y)$ in \eqref{clf1}. Open-loop trajectories for the systems \eqref{pp-sf} and \eqref{pp-both} for $U=1$.}
\label{fig:open-loop}
\end{figure}

This paper is not about designing the best possible controller. One can find a study at designing various controllers and understanding the tradeoffs among them in \cite{11080060}. The paper is, instead, about taking the simplest controller that does the job of global stabilization, and uses positive harvesting only, and finding the corresponding strict CLF. 

What would be the ``simplest controller'' that does the job? We take that to be a feedback that employs just linear feedback of the predator concentration. For the model  \eqref{pp-sf} we take the simplest controller as
\begin{equation}
\label{contro-sf}
U=Y
\end{equation}
and for the model 
\eqref{pp-both}
\begin{equation}
\label{control-both}
U=1+\varepsilon(Y-1)\,, \qquad \varepsilon\in(0,1)
\end{equation}
where $\varepsilon\in(0,1)$ ensures that $U>0$. The equilibrium value for both controllers is $U^*=1$, achieving $X^*=Y^*=1$. 

It is probably clear to the reader that we are taking nearly all the coefficients of the plant and the controller as unity. This is on purpose. We want to offer results in which the nonlinear structure is put on display as clearly as possible, uncontaminated by half a dozen various positive coefficients. A skilled reader won't have a problem to generalize our unity-coefficient results. 

Throughout the paper we use the 
% monotonically increasing function
% \begin{equation}
% \label{phi}
% \phi(s) = {\rm e}^s -1
% \end{equation}
standard Volterra Lyapunov function building block 
\begin{eqnarray}
\label{eq-Volterra-Lyapunov}
\Psi(S) &=& S-1 - \ln (S)
\nonumber\\
&=&(S-1)^2\underbrace{\displaystyle\int_{0}^{1}\frac{1-\theta}{\bigl(1+\theta(S-1)\bigr)^{2}} \, d\theta}_{>0}\,,
\end{eqnarray}
which is positive definite at $S=0$ and radially unbounded over $S>0$. 
% and its positive definite, radially unbounded derivative 
% \begin{equation}
% \Phi(s) = \int_0^s \phi(\sigma) d\sigma = \phi(s) - s\,.
% \end{equation}

\section{CLF for Predator-Prey Model with Predator Harvesting Only}
\label{sec-sf}

Before presenting a strict CLF, we note that, along the solutions of  \eqref{pp-sf}, \eqref{contro-sf}, the function 
\begin{equation}
\label{eq-V1-sf}
V_1(X,Y) = \Psi(X) + \Psi(Y)\,,
\end{equation}
which is our {\em non-strict} CLF \cite[(2)] {malisoff2016stabilization}, has a derivative
\begin{equation}
\dot V_1 = - (Y-1)^2
\end{equation}
and global asymptotic stability at $X=Y=1$ follows from the Barbashin-Krasovskii theorem, with the entire quadrant $\{X,Y>0\}$ as the region of attraction. 

At the risk of stating the obvious, both the non-strict CLF \eqref{eq-V1-sf} and the strict CLFs that follow throughout the note are in fact {\em (asymmetric) barrier Lyapunov functions},  for a system whose state space is $\mathbb{R}_+^2$ and the equilibrium is $X=Y=1$. 

\begin{figure}[t]
\centering
\includegraphics[width=.7\columnwidth]{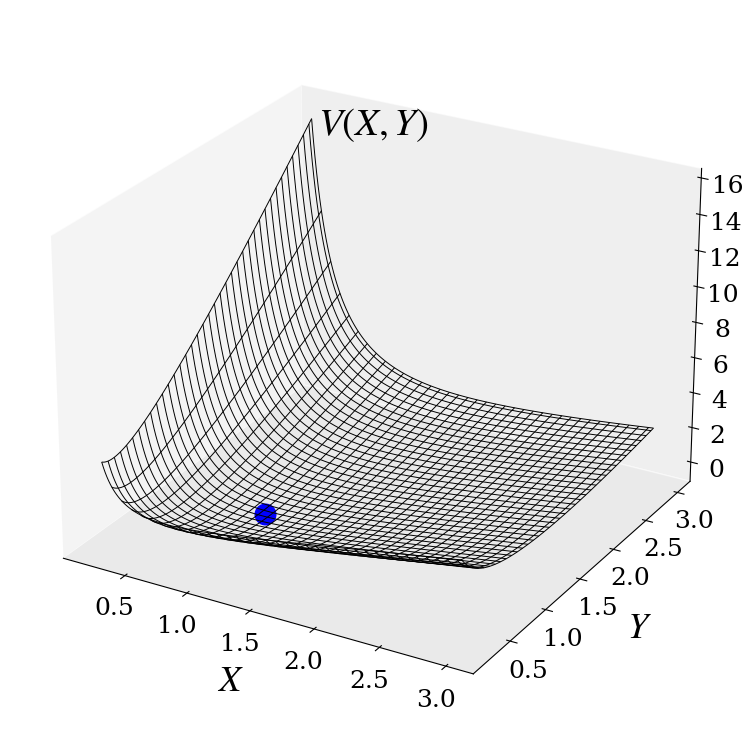}
\caption{$V(X,Y)$ in \eqref{clf1}. Positive definite and radially unbounded. Blue dot: equilibrium at $X=Y=1$.}
\label{fig:V-sf}
\end{figure}

\begin{theorem}
\label{thm-sf}
For the system \eqref{pp-sf}, the function
\begin{equation}
V=V_1+V_2\,,
\end{equation}
where
\begin{equation}
\label{V2}
V_2(X,Y) = \Psi(1/X) + \Psi(Y/X)\,,
\end{equation}
namely,
\begin{equation}
\label{clf1}
V={\color{blue}\frac{(X-1)^2}{X }}+ Y-1 - \ln(Y) + {\color{blue}\frac{Y}{X}-1 -\ln\left(\frac{Y}{X}\right)}
\end{equation}
is a strict CLF over $\{X>0,Y>0\}$.
\end{theorem}

\begin{figure}[t]
\centering
\includegraphics[width=.85\columnwidth]{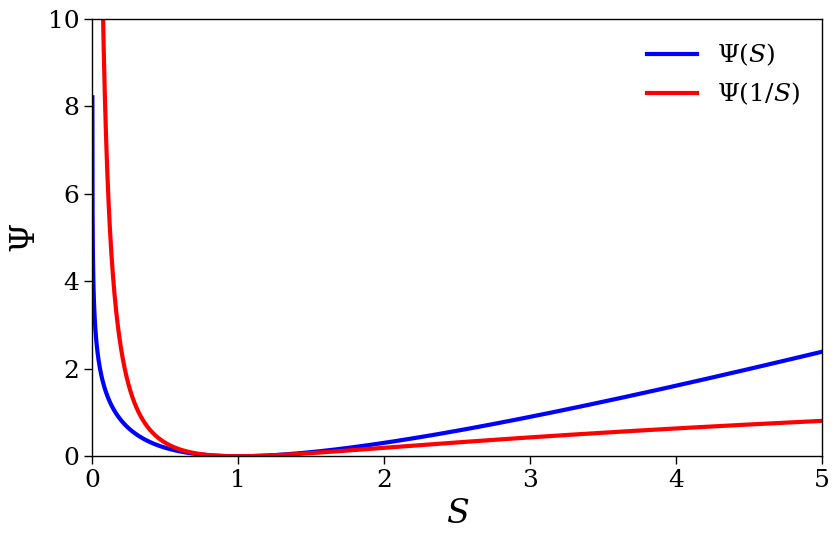}
\caption{The barrier Lyapunov functions $\Psi(S) = S-1-\ln(S)$ (blue) and $\Psi(1/S) = 1/S - 1+ \ln(S)$ (red). Both are highly asymmetric but the red dominates on the interval $(0,1)$ at the rate ($1/S$, whereas the blue dominates on $(1,\infty)$ at the rate $S$. The logarithmic growth is always weaker---both near $S=0$ and for large $S$.}
\label{fig:Psi}
\end{figure}

\begin{proof}
The positive definiteness and the radial unboundedness of $V(X,Y)$ over $\{X>0,Y>0\}$ are straightforward. \footnote{And illustrated in Figure \ref{fig:V-sf}.} The derivative
\begin{equation}
\dot V = - \frac{(X-1)^2}{X }- (Y-1)^2
\end{equation}
is negative definite over the entire set. Hence $V$ is strict. 
\end{proof}

Since, in the Lyapunov candidates \eqref{eq-V1-sf} and \eqref{V2}, both the barrier Lyapunov functions $\Psi(X)$ and $\Psi(1/X)$ appear, it is useful to first see, plainly, in what manner they are both positive definite and radially unbounded, and, second, to  note the difference in their growth rates and the forms of asymmetry relative to $X=1$. This is all displayed in Figure \ref{fig:Psi}.

The difference in simplicity between the CLF \eqref{clf1} and the CLF \cite[(12)]{malisoff2016stabilization} is striking.

Analyses of population dynamics are often conducted in the variables
\begin{equation}
    x=\ln(X)\,, \quad y=\ln(Y)
\end{equation}
with the model \eqref{pp-sf} replaced by
\begin{subequations}
\label{pp-sf-ln}
\begin{eqnarray}
\label{pp-sf1-ln}
\dot x &=& -\phi(y)
\\
\label{pp-sf2-ln}
\dot y &=& \phi(x) - u
\end{eqnarray}
\end{subequations}
where $u=U-1=\phi(y)$ and
\begin{equation}
\label{phi}
\phi(s) = {\rm e}^s -1\,.
\end{equation}
Defining
\begin{equation}
\Phi(s) = \int_0^s \phi(\sigma) d\sigma = \phi(s) - s\,,
\end{equation}
the CLF $V$ is rewritten as
\begin{equation}
\label{clf1-ln}
V = \Phi(x)+\Phi(-x) + \Phi(y) +\Phi(y-x)\,. 
\end{equation}

The high asymmetry of the dependence of $V$ on both the small and large values of $X$ and $Y$, and on $X$ versus $Y$ is evident both from \eqref{clf1} and \eqref{clf1-ln}, as well as in Figure \ref{fig:V-sf}. The sensitivity of $V(X,Y)$ is the highest to the prey concentration $X$ near prey extinction---a finite distance from equilibrium; the lowest to the harvested predator concentration $Y$, when the predator population is abundant and unlimited in possible further growth. This asymmetry in $V $ also reflects the asymmetry in only the prey being harvested and measured.  

In the language of \cite{9740520}, the CLF \eqref{clf1} is clearly not ``separable,'' given its term $Y/X$. 

One wants to note that  predator-to-prey ratio $Y/X$ in   \eqref{clf1}, as well as $y-x$ in \eqref{clf1-ln}, is a form of backstepping transformation, of the predator concentration modified by the prey concentration. This is a significant difference from the strictification approach in  \cite{9740520,ItoMalMaz}, as well as the Matrosov-based approach in \cite{malisoff2016stabilization}. 

\section{CLF for Simultaneous  Harvesting of Predator and Prey}
\label{sec-both}

We consider the dynamics \eqref{pp-both} with the feedback \eqref{control-both} and $\varepsilon\in(0,1)$, namely, 
\begin{subequations}
\label{pp-both-closed}
\begin{eqnarray}
\label{pp-both1-closed}
\dot X &=& -(1+\varepsilon)(Y-1)X
\\
\label{pp-both2-closed}
\dot Y &=& [X-1-\varepsilon(Y-1)]Y\,,
\end{eqnarray}
\end{subequations}
This is a more challenging system for a strict CLF construction. A non-strict CLF is somewhat simple,
\begin{equation}
V_1(X,Y) = \Psi(X) + (1+\varepsilon) \Psi(Y)\,
\end{equation}
and has a derivative
\begin{equation}
\dot V_1 = -(1+\varepsilon)\varepsilon(Y-1)^2\,,
\end{equation}
which is negative semidefinite. 

\begin{theorem}
\label{thm-both}
For the system \eqref{pp-both}, the function
\begin{equation}
V=V_1+V_2\,,
\end{equation}
where
\begin{equation}
\label{V2-both}
V_2(X,Y) = \Pi(X)+\Psi\left(\frac{Y}{X^{\frac{\varepsilon}{1+\varepsilon}}}\right)\,,
\end{equation}
where
\begin{eqnarray}
\label{Pi(X)}
\Pi(X) &=&\frac{1}{X^{\frac{\varepsilon}{1+\varepsilon}}}\left[ X-1-\frac{1+\varepsilon}{\varepsilon} \left(X^{\frac{\varepsilon}{1+\varepsilon}}-1\right) \right] 
\nonumber \\
&=&(X-1)^2\underbrace{\displaystyle\int_{0}^{1}\frac{1-\theta}{\bigl(1+\theta(X-1)\bigr)^{\frac{1+2\varepsilon}{1+\varepsilon}}}d\theta}_{>0}\,,\qquad
\end{eqnarray}
namely,
\begin{eqnarray}
\label{clf2}
V &=& (X-1)-\ln X
+(1+\varepsilon)\big[(Y-1)-\ln Y\big] 
\nonumber\\ &&
+{\color{blue}\frac{1}{X^{\frac{\varepsilon}{1+\varepsilon}}}
\left[
X-1-\frac{1+\varepsilon}{\varepsilon}
\left(X^{\frac{\varepsilon}{1+\varepsilon}}-1\right)
\right] }
\nonumber\\ &&
+{\color{blue}\left[
\frac{Y}{X^{\frac{\varepsilon}{1+\varepsilon}}}-1
-\ln\left(\frac{Y}{X^{\frac{\varepsilon}{1+\varepsilon}}}\right)
\right]}
\end{eqnarray}
is a strict CLF over $\{X>0,Y>0\}$.
\end{theorem}

\begin{proof}
The positive definiteness and radial unboundedness of $V(X,Y)$ (illustrated in Figure \ref{fig:V-both}) is evident from those properties being held by $V_1$ and from the nonnegativity of $\Pi(X)$ and $\Psi\left({Y}/{X^{\frac{\varepsilon}{1+\varepsilon}}}\right)$. Interestingly, $\Pi$ is positive definite, radially unbounded, but not convex. The derivative of $V$, 
\begin{equation}
\dot V = -\frac{(X-1)\Bigl(X^{\frac{\varepsilon}{1+\varepsilon}}-1\Bigr) }{X^{\frac{\varepsilon}{1+\varepsilon}}}-(1+\varepsilon)\varepsilon(Y-1)^2
\end{equation}
is  evidently negative definite.
\end{proof}

\begin{figure}[t]
\centering
\includegraphics[width=.7\columnwidth]{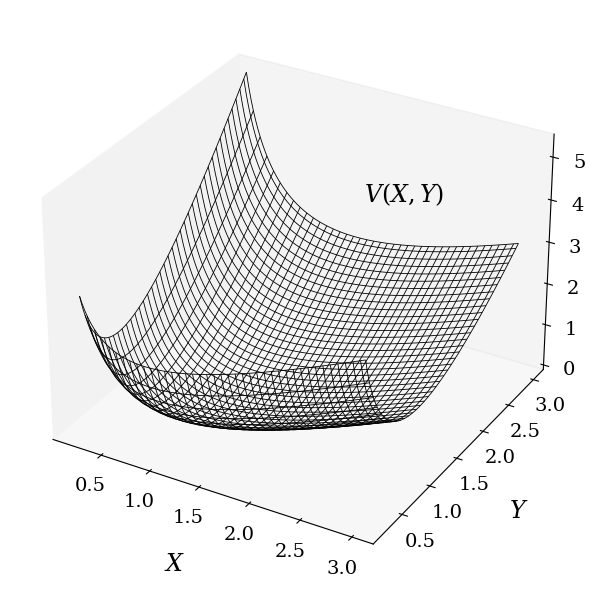}
\caption{$V(X,Y)$ in \eqref{clf2} for $\varepsilon=0.5$.}
\label{fig:V-both}
\end{figure}

%\newpage

The CLF component $\Pi(X)$ in \eqref{V2-both} is constructed in a systematic manner. After computing
\begin{eqnarray}
\dot V_2
&=&(Y-1)\left[-(1+\varepsilon)X\Pi'(X)+\frac{X-1}{X^\alpha}\right]
\nonumber\\
&&
\underbrace{\displaystyle -
%\left(
\frac{(X-1)\left(X^\alpha-1\right)}{X^\alpha}}_{<0\,, \ \forall X\in(0,1)\cup (1,\infty)}
%-1\right)
\,,\qquad
\alpha=\frac{\varepsilon}{1+\varepsilon}\,,
\end{eqnarray}
we choose $\Pi(X)$ so that the indefinite cross-term that has a factor of $Y-1$ gets cancelled, namely, we choose
\begin{equation}
\Pi'(X)%=\frac{X-1}{(1+\varepsilon)\,X^{1+\alpha}}
=\frac{X-1}{(1+\varepsilon)\,X^{\,1+\frac{\varepsilon}{1+\varepsilon}}}\,,
\end{equation}
and after integrating it arrive at the positive definite expression \eqref{Pi(X)}. It is in the same systematic manner that the component $\Psi(1/X)$ in \eqref{V2} was arrived at.

\begin{figure}[t]
\centering
\includegraphics[width=.7\columnwidth]{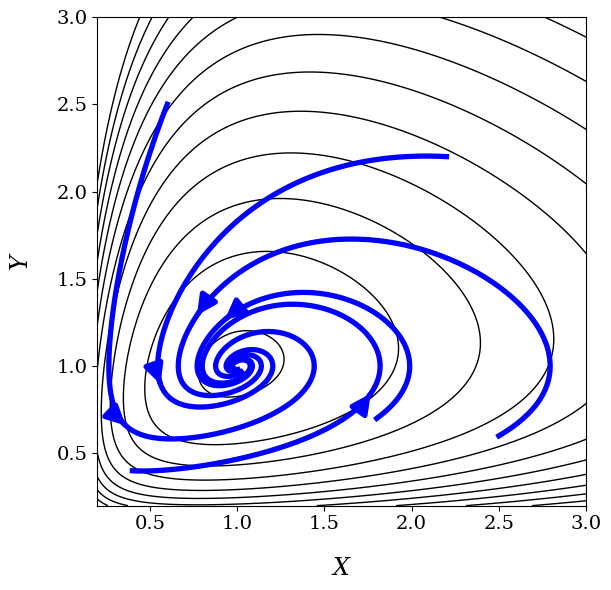}
\caption{Level sets of $V(X,Y)$ in \eqref{clf2} for $\varepsilon=0.9$ and trajectories of the closed-loop system \eqref{pp-both-closed}. The level sets of $V$ may only appear to be like the open-loop trajectories in Figure \ref{fig:open-loop} but a closer inspection reveals that they are definitely not.}
\label{fig:V-both-levels}
\end{figure}

As expected, the CLF \eqref{clf2} has level sets that get dense near the extinction state, as shown in Figure \ref{fig:V-both-levels}, which also shows some trajectories for $\varepsilon = 0.9.$

\section{Bonus: Forwarding-Based and Backstepping-Based Designs}
\label{sec-bonus}

In this section we go beyond the study of strictification of CLFs under given feedback laws. We explore new designs of feedback laws and strict CLFs for predator-prey dynamics under the restriction of positive harvesting.

Since feedback is designed along with a CLF, one has to differentiate between CLFs with unrestricted $U\in\mathbb{R}$ and CLFs with $U>0$. This section is dedicated to efforts towards the latter, harder-to-find CLFs.

\subsection{Forwarding CLF}

If one defines $x=\ln(X), \eta=-\ln(Y), u=U-X$, the plant \eqref{pp-sf} is rewritten as
\begin{subequations}
\label{pp-fw}
\begin{eqnarray}
\label{pp-fw}
\dot x &=& 1-{\rm e}^{-\eta} = -\phi(-\eta)
\\
\label{pp-fw}
\dot \eta &=&  u\,,
\end{eqnarray}
\end{subequations}
which is a ``nonlinear integrator chain,'' which belongs to the strict-feedforward class \cite{SEPULCHRE1997979}. There are numerous ways to (1) construct the forwarding transformation, (2) pick a target system, and (3) complete the Lyapunov analysis. After thoroughly exploring the options, we observe that the standard choices for systems whose states are real (rather than positive) and which are dominated by an integrator chain result in designs and analyses that, while mathematically impressive (employing special functions and sophisticated majorizations) are inelegant. The right forwarding approach for the predator-prey system exploits its nonlinearities and proceeds as follows. 

The forwarding transformation chosen as simple as 
\begin{equation}
\xi=x+\eta 
\end{equation}
results in
\begin{subequations} 
\begin{eqnarray}
\dot\xi  &=& -\phi(-\eta)+u,\\
\dot\eta &=& u.
\end{eqnarray}
\end{subequations}
Then we choose the forwarding feedback
\begin{equation}
u(\xi,\eta)=\phi(-\eta)-\phi(\xi),
\end{equation}
which yields the closed-loop system
\begin{subequations} 
\begin{eqnarray}
\dot\xi  &=& -\phi(\xi),\\
\dot\eta &=& \phi(-\eta)-\phi(\xi).
\end{eqnarray}
\end{subequations} 
For this system consider the positive definite radially unbounded Lyapunov function
\begin{equation}
V_0(\xi,\eta)=\Phi(\xi)+\Phi(-\eta)\,.
\end{equation}
Along closed-loop trajectories its derivative admits the completed-squares representation
\begin{equation}
\dot V_0
=
-\frac12
\Big[
\big(\phi(\xi)\big)^2
+
\big(\phi(-\eta)\big)^2
+
\big(\phi(\xi)-\phi(-\eta)\big)^2
\Big]\,,
\end{equation}
which is negative definite and radially unbounded. Hence, $V_0$ is a strict CLF on $\mathbb{R}^2$ (though the positivity of the associated feedback is yet to be examined). 

We translate this design and its conclusion into the following result.

\begin{theorem}
\label{thm3}
With the feedback law
\begin{equation}
\label{eq-Ufwd-simple}
U =  X+Y-\frac{X}{Y},
\end{equation}
the CLF
\begin{eqnarray}
\label{V-fwd-simple}
V(X,Y) &=& V_0(\ln(X/Y),\ln(1/Y))
\nonumber \\
&=&\frac{X}{Y}-1-\ln X+Y - 1\,,
\end{eqnarray}
which is positive definite at $X=Y=1$ and radially unbounded over $\{X,Y>0\}$, has the negative definite derivative
\begin{equation}
\label{eq-Vdot-fwd}
\dot V =-\frac12
\Bigg[\left(\frac{X}{Y}-1\right)^2+
(Y-1)^2 +\left(\frac{X}{Y}-Y\right)^2
\Bigg]
\end{equation}
along the solutions of \eqref{pp-sf}. 
\end{theorem}

\begin{proof}
Returning to the original variables, and noting that
\begin{eqnarray}
&\xi = \ln\!\left(\dfrac{X}{Y}\right)&
\\
&\phi(-\eta) = Y-1,\quad
\phi(\xi) = \dfrac{X}{Y}-1&
\\
&U = X+u\,, \quad u = Y-\dfrac{X}{Y}\,,&
\end{eqnarray}
the closed-loop system
\begin{subequations}
\begin{eqnarray}
\dot X &=& (1-Y)X\\
\dot Y &=& X-Y^2
\end{eqnarray}
\end{subequations}
yields \eqref{eq-Vdot-fwd}.
\end{proof}

\begin{figure}[t]
\centering
\includegraphics[width=.7\columnwidth]{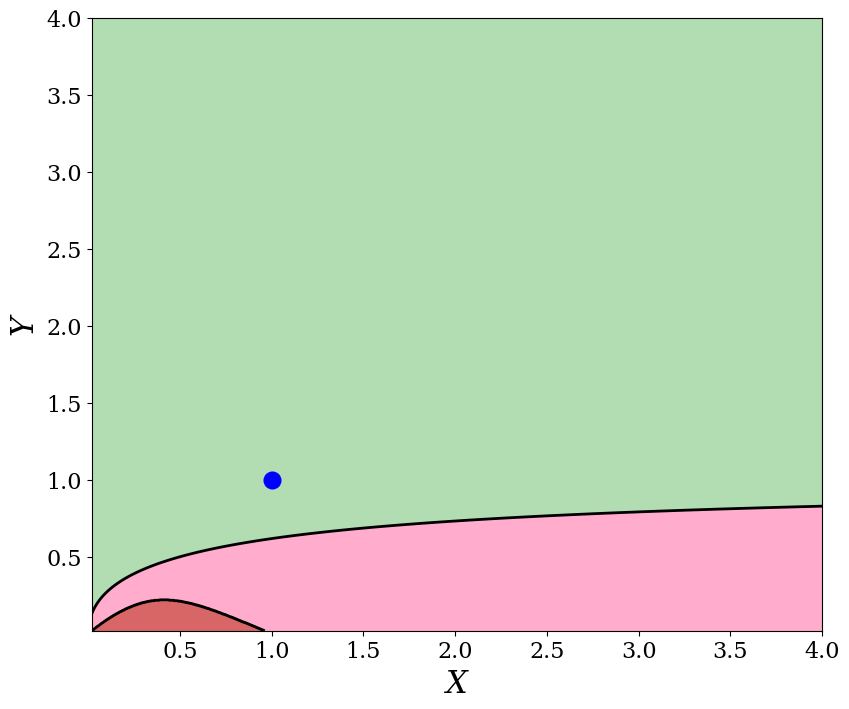}
\caption{%Feedback $U(X,Y)$ in \eqref{eq-U(X,Y)}. 
Green: $U>0$. Pink: $U<0$. Red: the set where both $L(X,Y)$ and $G(X,Y)$ in \eqref{eq-LG} are positive, namely, where no positive value of $U$ can make $\dot V$ in \eqref{eq-Vdot-U} negative, namely, where $V$ defined by \eqref{V-fwd-simple} fails to be an input-positive CLF. (Blue dot: equilibrium at $X=Y=1$.)}
\label{fig:U>0}
\end{figure}

Theorem \ref{thm3} is a nice mathematical result but, in practical terms, it is cautionary. The feedback \eqref{eq-Ufwd-simple}, which is more complex than the simple $U=Y$ in \eqref{contro-sf}, employs negative values of the harvesting {\em rate} when $X>Y^2/(1-Y)$, namely, when the predator is dominated by the prey. This set is shown in pink in Figure \ref{fig:U>0}. Note what $U<0$ represents: not the amount of introduction of the predator into the system, but the rate of addition. One could envision in some situations having a ``predator farm'' for the purpose of applying negative harvesting  for global stabilization, but this is seldom realistic. 

For the open-loop system \eqref{pp-sf} with arbitrary input \(U\), the derivative of \(V(X,Y)\) defined in \eqref{V-fwd-simple} can be written in the input-affine form
\begin{equation}
\label{eq-Vdot-U}
\dot V = L(X,Y)+G(X,Y)\,U,
\end{equation}
with
\begin{subequations}
\label{eq-LG}
\begin{eqnarray}
L(X,Y) &=& -\frac{X^2}{Y}+\frac{X}{Y}-X+Y-1+XY,\\
G(X,Y) &=& \frac{X}{Y}-Y.
\end{eqnarray}
\end{subequations}
The set in red in Figure \ref{fig:U>0} is the set where $\dot V$ can be made negative only with negative harvesting. Hence, while the feedback \eqref{eq-Ufwd-simple} is conservative in terms of the region attraction achieved with positive feedback, \eqref{V-fwd-simple} is simply not a global CLF if only positive $U$ is admissible.

The forwarding example in this section 
%aspect of Theorem \ref{thm3} 
illustrates how hard, even for the predator-prey system, let alone for more complex population dynamics, the problems of both  global stabilization and the strictification of CLFs are, due to the positivity constraint on control $U$. 

The nested saturation approach for feedforward systems might offer feedback laws that are both globally stabilizing and always positive valued. Alas, with such feedback laws one faces even greater challenges in constructing CLFs than with the integrator forwarding approach, and one typically relies on small-gain arguments rather than on CLF constructions \cite{536496}. 
%No clean CLF is known even for the predator-prey benchmark.  

Yet, we do not give up on expanding global strict CLF constructions and turn our attention next to backstepping. 

\subsection{Backstepping}

For predator-prey dynamics, we regard the predator-to-prey-surplus $Y/X$ as a backstepping transformation. Now, take the CLF \eqref{V2},
\begin{equation}
\label{V-bkst}
V(X,Y)
=\Psi\left(\frac{1}{X}\right)+\Psi\left(\frac{Y}{X}\right)\,,
\end{equation}
which was non-strict  in the context of the passivity-inspired feedback $U=Y$ in \eqref{contro-sf}, as well ass the feedback
\begin{equation}
\label{U-bkst<0}
U(X,Y)=Y+\frac{Y-X}{Y}\,,
\end{equation}
which can be regarded as a fortified version of. Then, along the solutions of \eqref{pp-sf} we have
\begin{equation}
\dot V = - \frac{(X-1)^2}{X}
-\frac{(Y-X)^2}{XY}\,.
%-\frac{(Y/X-1)^2}{Y/X}\,.
\end{equation}
While this $\dot V$ is negative definite, the backstepping feedback \eqref{U-bkst<0} employs negative harvesting for $X>(1+Y)Y$, i.e., when the predator is dominated by the prey. 

One might suspect that the backstepping CLF \eqref{V-bkst} is also hopeless like the CLF \eqref{V-fwd-simple} in terms of global validity for positive $U$. Fortunately, this pessimism is unfounded because the  open-loop dissipation dominates when $U> 0$. In open loop, the backstepping CLF \eqref{V-bkst} satisfies
\begin{equation}
\label{Vdot-with-L-G}
\dot V
= L(X,Y)+G(X,Y),U
\end{equation}
with 
\begin{subequations}
\begin{eqnarray}
\label{eq-L-def}
L(X,Y)&=&\frac{-(X-1)^2+Y(Y-X)}{X},
\\
\label{eq-G-def}
G(X,Y)&=&\frac{X-Y}{X}.
\end{eqnarray}
\end{subequations}
It is easily seen that when $G(X,Y)>0$, namely, when $X>Y$, the second term in the numerator of $L(X,Y)$, given by $Y(Y-X)$ is negative, in addition to the term $-(X-1)^2$ obviously being negative as well.  This means that the backstepping CLF candidate \eqref{V-bkst} is a global CLF with $U$ positively constrained. 

The question then is, while globally stabilizing positive feedback surely exists, can one find a feedback law that is simple and appealing? The affirmative answer is provided in the following theorem. 

\begin{theorem}
\label{thm5}
The feedback
\begin{equation}
\label{U-bkst<0+}
U(X,Y)=\frac{Y^{2}}{X}
\end{equation}
is a positive-valued global asymptotic stabilizer of the system \eqref{pp-sf} at $X=Y=1$ on the domain $\{X,Y>0\}$.\footnote{As illustrated by Figure \ref{fig:U>0}.}
\end{theorem}

\begin{proof}
In closed loop, the backstepping CLF \eqref{V-bkst} has a negative definite derivative
\begin{equation}
\dot V
=-\frac{(X-1)^2}{X}-\frac{(Y-X)^2}{X} \, \frac{Y}{X}\,.
\end{equation}
The positivity of $U(X,Y)={Y^{2}}/{X}$ is evident on the positive invariant set $\{X,Y>0\}$  and is inherent to the design, not a result of saturation or projection. 
\end{proof}

\begin{figure}[t]
\centering
\includegraphics[width=.9\columnwidth]{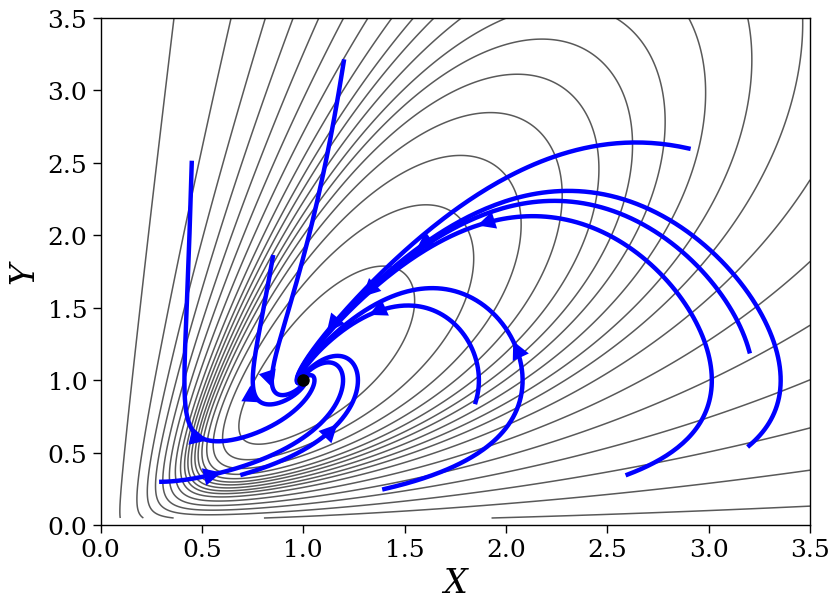}
\caption{Trajectories (blue) of the system \eqref{pp-sf} with the positive backstepping controller \eqref{U-bkst<0+}, overlaid on top of the level sets of the CLF \eqref{V-bkst}.}
\label{fig:U>0}
\end{figure}

It is important not to miss what has transpired here. The conventionally designed backstepping feedback \eqref{U-bkst<0} has to resort to negative harvesting for global stabilization. The unconventional, structure-exploiting backstepping feedback \eqref{U-bkst<0+} is globally stabilizing with $U>0$. 

This ``unconventional backstepping'' doesn't aim for a target system that is linear, with a negative diagonal and skew-symmetric off-diagonal coefficients. Instead, multiple the structure of the problem informs us to pursue  the entirely non-linear target system
\begin{subequations}
\label{target-bkst}
\begin{eqnarray}
\dot x &=& - \phi(x) +\frac{\phi(x)}{\phi(-x)}\phi(z)\\
\dot z &=& \phi(x) + {\rm e}^{y+z}\phi(-z)\,,
\end{eqnarray}
\end{subequations}
with highly asymmetric dissipative nonlinearities $\phi(x),\phi(-z)$. The CLF that corresponds to such a target system is not the usual sum of squares but the sum
\begin{equation}
V = \Phi(-x) +\Phi(z)\,,
\end{equation}
which is the alternative way of writing \eqref{V-bkst}, and results in
\begin{equation}
\dot V
=-4\left[\sinh^{2}\!\left(\frac{x}{2}\right)+e^{\,y+z}\sinh^{2}\!\left(\frac{z}{2}\right)\right]\,.
\end{equation}
In the $x,y,z$ variables, the feedback  \eqref{U-bkst<0+} is written as
\begin{equation}
U = {\rm e}^{y+z}
\end{equation}
and obviously positive. In summary, who would have guessed in initially facing the predator-prey system with backstepping that such a large deviation from this method's conventional version is what it takes to acheive global stabilization, with positive controls?

There is no general method here. The success of \eqref{U-bkst<0} is opportunistic. And, with a good CLF, some controllers may only partially succeed, while there may exist other controllers that succeed fully. Their discovery is up to the designer. As is the discovery of a CLF that is global and strict under the restriction to positive inputs. 

\section{Conclusions}

New CLF constructions, both for strictification under preselected feedbacks and done in parallel with feedback design, 
%for a population dynamics benchmark 
are introduced, for a benchmark model in population dynamics. 
%First, which comprise two non-strict Lyapunov constructions, to compose one that is strict, are  introduced in the first part of this note. 

{\bf Strictification Methodology.}  Sections \ref{sec-sf} and \ref{sec-both} present  constructions of second Lyapunov function, which exhibit dissipativity of the (prey) state whose dissipation is not captured in the first Lyapunov function. 
\begin{itemize}
    \item \underline{Step 1.} A  backstepping-inspired change of coordinates, the predator-to-prey ratio $Y/X$, is proposed for the strictification.
    \item \underline{Step 2.} An augmentation of the Lyapunov function, dependent on the prey state $X$ alone, and given specifically by the functions $\Psi(1/X)$ in Theorem \ref{thm-sf} and the more complex $\Pi(X)$ in Theorem \ref{thm-both}, is crucial component to be constructed for making the candidate $V_2(X,Y)$ an actual Lyapunov function and for yielding strictness. 
\end{itemize}

{\bf CLF Design.} Strictification is a process in which a feedback has already been preselected and the quest for a better CLF is only a matter of analysis. We  explore (in Sec. \ref{sec-bonus}) also a more ambitious goal of concurrently seeking strict CLFs and positive stabilizing feedbacks. In this larger design space, CLF design informed by classical methods like forwarding and backstepping, when undertaken in a manner customized to systems with positive states, might succeed. But success with the input  maintained positive is much harder. We have presented one successful CLF+positive-feedback construction (Thm. \ref{thm5}), from among multiple attempts. Some designs that succeed at global stabilization but fail at achieving positivity of control across the entire state space are shown in Sec. \ref{sec-bonus}, to illustrate the challenge. 

If these ideas do turn out to carry general potential, they may permit extensions, within the realm of control of population dynamics, along the lines of Karafyllis's recent work \cite{karafyllis2025resultsglobalattractivityinterior}, where his generalizations are still thus far reliant on classical (separable) Volterra-Lyapunov (linear+log) ingredients. With CLF strictness a gateway to robustness, in a hopeful scenario the proposed approach may lead from strict CLFs to ISS-CLFs, which is at present achieved, for particular classes of population dynamics, in notable results in \cite{malisoff2016stabilization,9740520,ItoMalMaz}. 

Our results are not universal but they are structural, not limited to population dynamics. Designers for positive systems, bilinear in a positive input, and with severely asymmetric dissipation, might be able to leverage our ideas.

\section*{References}
\bibliography{sample.bib}

@ARTICLE{536496,
  author={Teel, A.R.},
  journal={IEEE Transactions on Automatic Control}, 
  title={A nonlinear small gain theorem for the analysis of control systems with saturation}, 
  year={1996},
  volume={41},
  number={9},
  pages={1256-1270},
  doi={10.1109/9.536496}}

@article{SEPULCHRE1997979,
title = {Integrator forwarding: A new recursive nonlinear robust design},
journal = {Automatica},
volume = {33},
number = {5},
pages = {979-984},
year = {1997},
issn = {0005-1098},
doi = {https://doi.org/10.1016/S0005-1098(96)00249-X},
author = {Rodolphe Sepulchre and Mrdjan Jankovic and Petar V. Kokotovic}
}

@ARTICLE{9744516,
  author={Karafyllis, Iasson and Krstic, Miroslav},
  journal={IEEE Transactions on Automatic Control}, 
  title={Spill-Free Transfer and Stabilization of Viscous Liquid}, 
  year={2022},
  volume={67},
  number={9},
  pages={4585-4597},
  doi={10.1109/TAC.2022.3162551}}

@misc{karafyllis2025resultsglobalattractivityinterior,
      title={Results for Global Attractivity of Interior Equilibrium Points for Lotka-Volterra Systems}, 
      author={Iasson Karafyllis},
      year={2025},
      eprint={2512.11384},
      archivePrefix={arXiv},
      primaryClass={math.DS},
      url={https://arxiv.org/abs/2512.11384}, 
}

@ARTICLE{ItoMalMaz,
author = {Hiroshi Ito and Michael Malisoff and Frédéric Mazenc},
title = {Strict Lyapunov functions and feedback controls for SIR models with quarantine and vaccination},
journal = {Discrete and Continuous Dynamical Systems - B},
volume = {27},
number = {12},
pages = {6969-6988},
year = {2022},
issn = {1531-3492},
doi = {10.3934/dcdsb.2022029}
}

@ARTICLE{9740520,
  author={Ito, Hiroshi},
  journal={IEEE Transactions on Automatic Control}, 
  title={Strict Smooth Lyapunov Functions and Vaccination Control of the SIR Model Certified by ISS}, 
  year={2022},
  volume={67},
  number={9},
  pages={4514-4528},
  doi={10.1109/TAC.2022.3161395}}

@ARTICLE{11080060,
  author={Veil, Carina and Krstić, Miroslav and Karafyllis, Iasson and Diagne, Mamadou and Sawodny, Oliver},
  journal={IEEE Transactions on Automatic Control}, 
  title={Stabilization of Predator–Prey Age-Structured Hyperbolic PDE When Harvesting Both Species is Inevitable}, 
  year={2026},
  volume={71},
  number={1},
  pages={123-138},
  doi={10.1109/TAC.2025.3589108}}

@article{malisoff2016stabilization,
  title={Stabilization and robustness analysis for a chain of exponential integrators using strict Lyapunov functions},
  author={Malisoff, Michael and Krstic, Miroslav},
  journal={Automatica},
  volume={68},
  pages={184--193},
  year={2016},
  publisher={Elsevier}
}
\bibliographystyle{IEEEtranS}

\end{document}